%% file: main.tex
\documentclass[12pt]{article}
\usepackage{graphicx}
\graphicspath{{../figures/}{./figures/}}
\usepackage{fullpage}
\usepackage[margin=1.0in]{geometry}
\usepackage{amsthm,amsmath,amssymb}
\usepackage[T1]{fontenc}
\usepackage[utf8]{inputenc}
\usepackage{xurl}
\usepackage{natbib}
\usepackage[colorlinks,breaklinks=true]{hyperref}
\usepackage[dvipsnames]{xcolor}
\usepackage{enumitem}
\usepackage{mathtools}
\usepackage[nameinlink,noabbrev,capitalise]{cleveref}
\usepackage[mono=false]{libertinus}
\usepackage[libertine]{newtxmath}
\usepackage{anyfontsize}
\usepackage{algorithm}
\usepackage{algpseudocode}
\usepackage{setspace}

\hypersetup{
  colorlinks=true,
  linkcolor=Maroon,
  citecolor=NavyBlue,
  urlcolor=NavyBlue
}
\allowdisplaybreaks
\newtheorem{thm}{Theorem}
\newtheorem{lem}{Lemma}
\newtheorem{prop}{Proposition}
\newtheorem{coro}{Corollary}

\theoremstyle{definition}

\newtheorem{example}{Example}

\theoremstyle{remark}
\newtheorem{remark}{Remark}

\crefname{thm}{Theorem}{Theorems}
\crefname{lem}{Lemma}{Lemmas}
\crefname{prop}{Proposition}{Propositions}
\crefname{coro}{Corollary}{Corollaries}
\crefname{defn}{Definition}{Definitions}
\crefname{example}{Example}{Examples}
\crefname{assump}{Assumption}{Assumptions}
\crefname{remark}{Remark}{Remarks}

\numberwithin{equation}{section}
\crefformat{equation}{(#2#1#3)}
\Crefformat{equation}{(#2#1#3)}
\crefrangeformat{equation}{(#3#1#4)--(#5#2#6)}
\Crefrangeformat{equation}{(#3#1#4)--(#5#2#6)}
\crefmultiformat{equation}{(#2#1#3)}{ and (#2#1#3)}{, (#2#1#3)}{, and (#2#1#3)}
\Crefmultiformat{equation}{(#2#1#3)}{ and (#2#1#3)}{, (#2#1#3)}{, and (#2#1#3)}
 
\newcommand{\R}{\mathbb{R}}

\newcommand{\E}{\mathbb{E}}

\newcommand{\bigO}{O}
\newcommand{\smallo}{o}

\newcommand{\marginP}{P_\mathrm{margin}}
\newcommand{\marginp}{p_\mathrm{margin}}
\newcommand{\LVR}{\mathrm{LVR}}
\newcommand{\TE}{\mathrm{TE}}
\newcommand{\TC}{\mathrm{TC}}
\newcommand{\PnL}{\mathrm{PnL}}

\newcommand{\beq}{
\begin{equation}}
  \newcommand{\eeq}{
\end{equation}}
\newcommand{\beqs}{\beq
  \begin{aligned}
  }
  \newcommand{\eeqs}{
  \end{aligned}
  \eeq
} 

\algnewcommand\Assert{\textbf{assert }}
\algrenewcommand{\algorithmiccomment}[1]{\hfill\texttt{// #1}}

\title{Partially Active Automated Market Makers}
\author{
  Sunghun Ko\footnote{
    \texttt{kosunghun317@kaist.ac.kr}, College of Business, KAIST.
  }
  \and
  Seungwon Jeong\footnote{
    \texttt{seungwon.jeong@kaist.ac.kr}, College of Business, KAIST.
  }
}
\date{\today}

\begin{document}

\maketitle

\begin{abstract}
  We introduce a new class of automated market maker (AMM), the \emph{partially active automated market maker} (PA-AMM).
  PA-AMM divides its reserves into active and passive parts and uses only the active part for trading.
  At the start of each block, such division is recalculated to keep active reserves a fixed proportion of total reserves.
  We show that this simple mechanism reduces adverse-selection costs, measured by loss-versus-rebalancing (LVR), relative to plain constant-function market makers (CFMMs).
  As a trade-off, the asset weights within a PA-AMM pool may deviate from their target weights implied by its invariant curve.
  Motivated by the literature on optimal index-tracking problems, we also propose and solve an optimization problem that balances this deviation against reduced LVR.
\end{abstract}

\input{intro}
\input{preliminaries}
\input{pa_amm}
\input{optimization}
\input{experiments}
\input{conclusion}
\bibliography{references}
\bibliographystyle{chicago} 
\input{appendix}

\end{document}

%% file: intro.tex
\section{Introduction}


Automated market makers (AMMs) have gained popularity and become a dominant type of decentralized exchange (DEX) on blockchains due to their simplicity and low computational and operational costs.
However, because of their passive nature, AMM quotes cannot be updated in real time, and liquidity providers (LPs) suffer from adverse selection.
More specifically, since asset prices are often discovered on centralized exchanges (CEXs) and AMMs lack access to off-chain information about the true value of assets, arbitrageurs can exploit stale AMM quotes and take profit.
Such profit is estimated to exceed $\mathdollar 230$ million between August 2023 and March 2025 \citep{wu2025measuring}.

There has been extensive work on the structure and magnitude of such losses, both theoretically and empirically.
In early attempts, practitioners used the notion of ``impermanent loss'' (sometimes called ``divergence loss'') as a cost of liquidity provision \citep{pintail2019understanding_uniswap_returns,pintail2019uniswap_good_deal,aigner2021uniswap_il_riskprofile,deribit2020_il_bancorv2}.
In their seminal work, Milionis et al. proposed a notion of \emph{loss-versus-rebalancing} (LVR), analogous to markout in the sense that they both measure the quality of each trade with the price shortly or immediately after the trade, as a metric to measure the loss (i.e., the cost) of liquidity providers (LPs) of AMMs from adverse selection \citep{milionis2022automated}.
It has become a go-to metric for measuring loss from CEX-DEX arbitrage.
Initially, its closed-form formula was based on assumptions of an ideal situation (continuous-time model, no friction in arbitrageurs' trades, no gas fees), but such assumptions were relaxed and further generalized in subsequent works \citep{milionis2024automated,nezlobin2025loss,schlegel2025arbitrage}.
The notion of \emph{predictable loss} (PL), a generalization that consists of a convexity cost and an opportunity cost, was also proposed in \cite{cartea2023predictable}.
With additional assumptions such as zero opportunity cost and asset prices following geometric Brownian motion with constant volatility, PL reduces to LVR.

Moreover, studies on optimal liquidity provision, which concern selecting the fee rate and liquidity-provision range to balance fee revenue from noise traders and losses to arbitrageurs \citep{evans2021optimal,cartea2024decentralized,powers2024tick,alexander2024fees,ma2024cost,volosnikov2024impact,campbell2025optimal}, have emerged.
Especially, these studies collectively indicate that the optimal choice is determined by various factors, such as the arrival rate and distribution of order size of noise traders, asset volatility, liquidity depth of the competing venue, gas fee costs, and the time interval between two consecutive blocks of the underlying chain (block time).

Several mechanisms have also been proposed to mitigate LVR and protect LPs.
The proposed mechanisms fall largely into three categories: batch execution of trades \citep{ramseyer2024augmenting,canidio2023arbitrageurs}, dynamic fee rates \citep{CrocSwap2022Designing,cartea2026strategic,egorov2021dynamicpeg,baggiani2025optimal}, and discriminatory quotes based on (inferred) toxicity of orders \citep{bukov2020mooni,adams2024amm,solarcurve2023bip295}.
These approaches can be effective, but they often require additional infrastructure (e.g., an off-chain batcher, an administrator, or specialized auction mechanisms), rely on timely and robust information updates (e.g., price oracles and/or inference on the toxicity of order flow), and may reduce atomic composability or increase implementation complexity.

There is another lens for viewing AMMs: as a portfolio of assets that rebalances at each block to keep the capital allocation to each asset close to its target weight \citep*{angeris2021replicating,angeris2023replicating}.\footnote{
  Note that the target weight of each asset need not be kept fixed.
  For instance, asset weights in a Uniswap V3 LP position may vary as prices move within a finite liquidity range \citep{uniswap2021whitepaper}.
  Balancer's liquidity bootstrapping pool (LBP) is an example of an AMM with time-varying weights \citep{balancer_lbp}.
}
When relative prices change, arbitrageurs trade against the pool, pushing reserves back toward the invariant-implied composition, so the AMM rebalances indirectly through arbitrage rather than explicit order submissions.

Within such a framework, constant function market makers (CFMMs) and the aforementioned improved mechanisms operate suboptimally in the sense that they immediately remove any deviation beyond a certain threshold from the target weights; that is, if the arbitrage happens in the current block and the true price of the asset remains the same in the next block, no trade from arbitrageurs will happen even with zero friction since the opportunity has already gone in the last block.
In the optimal portfolio management literature, including index-tracking and optimal execution, this is rarely optimal.
An optimal policy often comprises two components: a no-trading-zone and an execution plan, which it follows when a weight discrepancy escapes that no-trading-zone \citep{buckleykorn1998,pliskasuzuki2004,liuzheng2016,benidis2018optimization,donohue2003optimal,holden2013optimal,garleanu2013dynamic,garleanu2016dynamic}.
The optimal execution plan is usually not a single large market order filled immediately, but a series of sub-orders executed over time \citep{AlmgrenChriss2001OptimalExecution,obizhaeva2013optimal}.
Except for \cite{fritsch2024mev,bukov2020mooni}, to our knowledge, there is no study on the mechanisms by which AMMs perform such \emph{gradual} alignment of the prices of assets on AMMs with their true prices.



In this paper, we introduce a new class of AMMs, \emph{partially active automated market makers} (PA-AMMs), consisting of an underlying invariant function \(\varphi\) and the \emph{activeness} parameter \(\lambda \in (0,1]\).
The special case of $\lambda=1$ recovers the standard CFMM with invariant function \(\varphi\).
Unlike most other proposed mechanisms to mitigate LVR, our mechanism requires no additional off-chain infrastructure, such as oracles or administrators, and, because of its simplicity, can be used alongside other mechanisms.

We demonstrate that determining the rebalancing speed is equivalent to determining \(\lambda\), the fraction of liquidity available in a single block.
We also investigate its properties, such as how it reduces the LVR as \(\lambda\) decreases, and, as a trade-off, how the tracking error, defined as the deviation of asset allocations from their target weights, increases.
Throughout the paper, we formalize this trade-off and define and solve the optimization problem to determine the optimal level of activeness, \(\lambda^\ast\), that balances it.
Lastly, we present numerical simulation results demonstrating the effectiveness of the proposed mechanism, even after accounting for noise traders and the fees they incur.
To our knowledge, this is the first study in the AMM and DeFi literature to explicitly and systematically examine the importance of the execution-speed perspective on rebalancing policy for AMM performance.

The paper is organized as follows.
\Cref{sec:prelim} reviews concepts and notation around CFMMs and the underlying model setup that our analysis is based upon.
\Cref{sec:pa_amm} defines PA-AMMs and derives the gap dynamics, stationary behavior, and the implied LVR rate.
\Cref{sec:opt} formulates and solves the optimal-activeness problem.
\Cref{sec:exp} measures the performance of PA-AMM via hedged PnL (profit and loss) across various combinations of swap fee rates and activeness via Monte Carlo simulation, thereby partly addressing concerns about its performance in the presence of noise traders and fees.
In \cref{sec:conc}, we conclude with implications, limitations, and directions for future work.

%% file: preliminaries.tex
\section{Preliminaries}\label{sec:prelim}

In this section, we recall the preliminary concepts and notation on AMMs used throughout the paper.
Although our work extends straightforwardly to cases with more than two assets, we focus on the two-asset case for notational simplicity.
For a more detailed explanation, we refer readers to \cite{evans2020liquidity}, \cite{angeris2021analysis}, \cite{angeris2020improved}, \cite{angeris2023geometry}, and \cite{bichuch2026axioms}.

\subsection{Constant Function Market Makers}

Let the \emph{trading function} \(\varphi: \R^2_+ \to \R_+ \) be \(C^3\), concave, and \(1\)-homogeneous.\footnote{
A function \(f\) is \(\alpha\)-homogeneous if \(f(kx) = k^\alpha f(x)\) for all \(x\) in its domain.
} 
Then the \emph{constant function market maker} (CFMM) with reserves \(R = (x, y) \in \R^2_+ \) is a trader who accepts any trade (change of reserves) \(\Delta R\) as long as the invariant implied by its trading function \(\varphi\) holds, i.e., \(\varphi(R + \Delta R) \geq \varphi(R)\).
The value of \(\varphi\) for given \((x,y)\) is called the \emph{liquidity}.
\emph{The marginal (spot) price of asset \(X\) in terms of asset \(Y\)} is then given by \begin{equation}
\marginP(x,y) = \left. -\frac{dy}{dx} \right|_{R = (x, y)} = \frac{\varphi_x(x,y)}{\varphi_y(x,y)} \, .
\end{equation}
We will denote the \emph{log marginal price} by \(\marginp(x,y) = \log \marginP(x,y)\).
Given liquidity \(L\) and log marginal price \(p\), we also define the reparametrization \(R(L, p) = (x(L, p), y(L, p))\) of reserves which satisfies \begin{equation}
\varphi(R(L, p)) = L \quad \text{and} \quad \marginp(R(L, p)) = p \, .
\end{equation}
We further assume that \( \marginp \) and reparametrization \(R(\cdot,\cdot)\) are well-defined and differentiable with non-singular Jacobian almost everywhere.
We will often abbreviate to \(x(p)\) and \(y(p)\) when it is clear that liquidity \(L\) remains constant in the context.
We also introduce a notion of \emph{pool wealth} \(W(L, P, S) \coloneqq S \, x(L,\log P) + y(L, \log P)\), which is the total value (in terms of \(Y\)) of assets within the pool with reserves \(R(L, \log P)\) when the true price of asset \(X\) (in terms of \(Y\)) is \(S\).
Again, when it is clear that \(L\) remains constant, we may abbreviate it as \(W(P, S)\) or even \(W(P)\) when \(P = S\).

\begin{example}
  \label{example:g3m_definition}
  A geometric mean market maker (G3M) is a CFMM with invariant function \(\varphi\) defined as \begin{equation}
  \varphi (x_1, x_2, \cdots, x_n) \coloneqq \prod^n_{i=1} x_i^{w_i},
  \end{equation} where \((w_i)_{i \in [n]}\) is the weight of the \(i\)-th asset satisfying \(\sum^n_{i=1} w_i = 1\).
  For a two-asset G3M with reserves \((x,y)\) and weights \((\theta, 1-\theta)\), the marginal price of asset \(X\) in terms of asset \(Y\) is then: \begin{equation}
  \marginP(x, y) = \frac{\theta}{1-\theta} \cdot \frac{y}{x},
  \end{equation} and for given liquidity \(L\) and log marginal price \(p\), the reserves \((x, y)\) are: \begin{equation}
  x(L, p) = L \left(\frac{\theta}{1-\theta} \cdot e^{-p} \right)^{1 - \theta}, \quad y(L, p) = L \left( \frac{1 - \theta}{\theta} \cdot e^p \right)^{\theta}.
  \end{equation}
  When the true price is \(S\), the pool wealth \(W(L,P,S)\) is: \begin{align}
  W(L,P,S) &= S \cdot x(L,\log P) + y(L,\log P) \\
  &= S \cdot L \left(\frac{\theta}{1 - \theta} \cdot \frac{1}{P} \right)^{1-\theta} + L \left( \frac{1 - \theta}{\theta} \cdot P \right)^{\theta} \\
  &= L \left( \frac{1-\theta}{\theta} \cdot P \right)^\theta \left(1 + \frac{\theta}{1-\theta} \cdot \frac{S}{P} \right) \\
  &= y(L,\log P) \cdot \left(1 + \frac{\theta}{1-\theta} \cdot \frac{S}{P} \right),
  \end{align} and for the case of \(P = S\), the weight of asset \(Y\) in the pool is indeed \(1 - \theta\).
\end{example}

\subsection{Setup}

Throughout the paper, we assume that the quote asset \(Y\) is a stablecoin and the base asset \(X\) is a risky asset whose (true) price follows geometric Brownian motion with drift \(\mu\) and volatility \(\sigma\) such that \(\mu = \sigma^2/2\).
That is, \(S_n\), the price of \(X\) at block \(n\), satisfies: \begin{equation}
\log S_n = \log S_{n-1} + \varepsilon_n,
\end{equation} where \(\varepsilon_n \sim \mathcal N( 0, \sigma^2 \Delta t), \, \forall n \in \mathbb N\) i.i.d. and \(\Delta t\) is the block time.\footnote{
  For technical reasons (the existence of a state-independent optimal policy for the optimization problem we investigate in \cref{sec:opt}), we restrict to this case.
  We note that allowing non-zero drift in the log-price evolution still does not change the leading-order term of the solution; it only affects the policy through the higher-order terms.
}
We also adopt the standard assumption of infinitely many arbitrageurs competing with one another to arbitrage against the AMMs, with no retail investors, for simplicity.
They are risk-neutral and myopic, caring only about maximizing the immediate profit available in a given block if they win the opportunity.
Lastly, we abstract away the gas costs and frictions in trading at centralized exchanges (CEXs), meaning that on CEXs, they can trade asset \(X\) for asset \(Y\) at \(S_n\) for any amount at block \(n\).
As a result, for CFMMs with no swap fee, the pool's spot price \(\marginP\) always coincides with the true market price after arbitrageurs' trades.

\subsection{Loss-versus-rebalancing}

Based on this setup, the loss of LPs to arbitrageurs, \emph{loss-versus-rebalancing} (LVR), is measured by the difference in the total value of assets within the pool before and after the arbitrage, i.e., \(W(P, S) - W(S)\).\footnote{
  This value is always greater than or equal to zero due to the concavity of \(\varphi\).
}
We recall the classical result introduced in \cite{milionis2022automated} where limit on the block time \(\Delta t \to 0\) is taken:
\begin{prop}[{\citealp[Theorem~1]{milionis2022automated}}]
  Loss-versus-rebalancing takes the form \begin{equation}
  \LVR_t = \int^t_0 \overline{\LVR}(P_s) \, ds, \quad \forall t \geq 0,
  \end{equation} where we define \(\overline{\LVR}(P)\), the instantaneous rate of LVR for \(P \geq 0\), by \begin{equation}
  \overline{\LVR}(P) \coloneqq -\frac{\sigma^2 P^2}{2} W''(P) \geq 0 \, .
  \end{equation}
  In particular, \(\LVR\) is a non-negative, non-decreasing, and predictable process.
\end{prop}

\begin{example}
  When the trading function is given by a weighted geometric mean, i.e., the CFMM is a geometric mean market maker (G3M) with weights \((\theta, 1 - \theta)\), the instantaneous rate of LVR becomes \begin{equation}
  \overline{\LVR}(P) = \frac{\sigma^2}{2} \theta (1 - \theta) W(P) \, .
  \end{equation}
\end{example}

This LVR rate for CFMMs will serve as a benchmark for measuring PA-AMM's performance in the subsequent section.\footnote{
  Note that while the model used in \cite{milionis2022automated} is based on continuous time, under the zero-swap-fee condition, the resulting LVR rate matches that of the fixed-block-time model \citep{nezlobin2025loss}.
}

%% file: pa_amm.tex
\section{The Partial Usage of Reserves}\label{sec:pa_amm} 
In this section, we introduce the partially active automated market maker (PA-AMM) and its properties.

\subsection{How It Works}

\paragraph{Top-of-the-block rebalancing and swap.}
The PA-AMM partitions the pool's total reserves $R_{\text{total}}$ into two distinct components: active reserves ($R_{\text{active}}$) and passive reserves ($R_{\text{passive}}$).
The underlying invariant function $\varphi$ enforces trading constraints solely on the active portion, thereby limiting the liquidity available for immediate arbitrage within any single block.
The mechanism ensures that the active/passive division occurs exactly once at the beginning of each block, triggered by the first interaction (swap, addition, or removal of liquidity) in that block.\footnote{
  The rebalancing period need not be fixed at 1 and can be generalized to any \(N \geq 1\).
}

More specifically, the rebalancing rule resets the active reserves to a fraction $\lambda \in (0, 1]$ of total reserves, while leaving the remaining $(1-\lambda)$ fraction in the passive reserves, which is unavailable for trading.
This occurs at the top of each block, and subsequent trades within the same block execute against the static $R_{\text{active}}$ until the next block increments.
This ensures that only a $\lambda$-fraction of the total liquidity is exposed to arbitrageurs at the top of each block.
The following example shows how rebalancing works and eventually reduces the LVR:

\begin{example}
  \label{ex:rebalancing_example}
  Consider a fee-free constant product market maker initially holding one unit each of assets \(X\) and \(Y\), i.e., \(R_\text{total} = (1,1)\), with invariant function \(\varphi(x,y) = \sqrt{xy}\). 
  At the beginning of block \(1\), the external price of \(X\), denominated in \(Y\), jumps from \(1\) to \( p = 1 + \delta\) and remains at \(p\) for the subsequent blocks.
  Suppose that a competitive arbitrageur interacts with the pool in every block following the jump.
  Let \((x_n,y_n)\) denote the total reserves immediately after arbitrage in the \(n\)-th block, and define \begin{equation}
    R_0 \coloneqq (x_0,y_0) = (1,1), \quad L_n \coloneqq \sqrt{x_n y_n}.
  \end{equation} 
  An arbitrage in the first block moves its reserves to \begin{equation}
    R_1
    = \left(x_1,y_1\right) 
    = \left(\frac{1}{\sqrt p},\sqrt p\right),
  \end{equation} after which they remain unchanged. 
  The corresponding arbitrage profit, denominated in units of \(Y\), is then \begin{equation}
    \Pi^{\mathrm{CFMM}} 
    = p+1- \left(p \cdot \frac{1}{\sqrt p}+\sqrt p\right)
    = \left(\sqrt p-1\right)^2
    \approx \frac{\delta^2}{4}
  \end{equation} for small \(\delta\).

  Now consider a PA-AMM with the same invariant function and activeness parameter \(\lambda \in (0,1)\).
  At the beginning of block \(n\), the reserves inherited from the preceding block are divided into active reserves \(\lambda(x_{n-1},y_{n-1})\) and passive reserves \( (1 - \lambda) (x_{n-1},y_{n-1})\).
  The active invariant is \begin{equation}
    \sqrt{(\lambda x_{n-1})(\lambda y_{n-1})}
    = \lambda L_{n-1}.
  \end{equation}
  Hence, after arbitrage aligns the marginal price of the active reserves with \(p\), we have \begin{equation}
    R_n^{\mathrm{active}}
    = \left(x_n^{\mathrm{active}},y_n^{\mathrm{active}}\right) 
    = \lambda L_{n-1} \left(\frac{1}{\sqrt p},\sqrt p\right).
  \end{equation}
  Adding back the unchanged passive reserves gives the recursion
  \begin{equation}
    \label{eq:pa_jump_recursion}
    R_n
    = (1-\lambda)R_{n-1} + \lambda L_{n-1} \left(\frac{1}{\sqrt p},\sqrt p\right), \quad n \geq 1,
  \end{equation} and iterating \eqref{eq:pa_jump_recursion} gives \begin{equation}
    \label{eq:pa_jump_unrolled}
    R_n
    = (1-\lambda)^nR_0 + \lambda \sum_{j=0}^{n-1} (1-\lambda)^{\,n-1-j}L_j \left(\frac{1}{\sqrt p},\sqrt p\right).
  \end{equation}
  The case \(\lambda=\frac{1}{2}\) admits a particularly simple closed form: \begin{equation}
    \label{eq:pa_jump_half_reserves}
      x_n
      = 2^{-n} \frac{(p-1)p^{\,2^{-n}}}{p\left(p^{\,2^{-n}}-1\right)}, 
      \quad
      y_n
      = 2^{-n} \frac{ p-1 }{ p^{\,2^{-n}}-1 }.
    \end{equation}
  Taking \(n\to\infty\) gives \begin{equation}
    (x_\infty,y_\infty)
    =
    \left(
      \frac{p-1}{p\log p},
      \frac{p-1}{\log p}
    \right),
  \end{equation} and hence the cumulative profit of arbitrageurs becomes \begin{equation}
    \Pi^{\mathrm{PA}}
    \to
    p+1-\frac{2(p-1)}{\log p}
    \approx \frac{\delta^2}{6}.
  \end{equation}
  Thus, the PA-AMM converges to the same marginal price \(p\) as the ordinary CFMM, but generally retains more liquidity and loses less value to arbitrageurs.
\end{example}

We remark that the rebalancing at most once per block is crucial.
Without such a limit, the trader can split her trade into many smaller trades and may obtain a better overall quote.
In particular, in the extreme case where rebalancing occurs after each trade and the trader splits the order into infinitely many orders of infinitesimal size, the quote given by PA-AMM becomes the same as that of CFMM.\footnote{
  A similar situation arises in \cite{canidio2023arbitrageurs}, making batching essential in their mechanism as well.
}

This observation, on the other hand, tells us that while the liquidity available within \emph{a single block} is only a \(\lambda\)-fraction of the corresponding CFMM with the same reserves, over \emph{multiple blocks}, both PA-AMM and CFMM provide the same amount of liquidity to traders, and this is what separates PA-AMM from simply depositing a \(\lambda\)-fraction of capital to CFMMs.

\paragraph{Addition and removal of liquidity.}
The addition and removal of liquidity are done in the following manner:
 Let $s_\text{total}$ be the total supply of LP shares token. 
 When a user tries to add (remove) liquidity, the rebalance check and execution occur first, as usual, followed by the addition (removal) operation. 
 The addition of liquidity requires the user to deposit \( (s_\text{new} / s_\text{total}) \cdot R_{\text{total}} \) of assets to mint \(s_\text{new}\) amount of LP shares, which will be distributed to the active and passive reserves proportionally.\footnote{
  Requiring deposits of both assets in the ratio of the current pool asset composition is standard in many implementations. Decentralized apps (Dapps) usually enable single-side liquidity provision via an auxiliary ``router'' contract (sometimes called a ``zapper''), which queries the pool state, calculates the required swap amount, and provides liquidity on behalf of users during transaction execution.
}
Similarly, removing liquidity decreases active and passive reserves in the same ratio, and both operations increase or decrease them proportionally, keeping their relative ratio static until the next rebalancing.\footnote{
  The reason for this design is clear: if we allow skewed liquidity deposits or withdrawals through a sequence of skewed liquidity addition and proportional withdrawal, advanced users can bypass the active/passive division and effectively swap against the passive reserves, which would defeat the purpose of our mechanism.
}

\subsection{The Dynamics of Price Gap}

We now derive the various properties of PA-AMM for the case of \( \varphi \) being a G3M with weights \((\theta, 1 - \theta)\), i.e., \(\varphi(x, y) = x^\theta y^{1-\theta}\) for \(\theta \in (0, 1)\).
Note that the results can be generalized to a broader class of CFMMs which satisfy certain conditions on the regularity of \(\varphi\).
The detailed conditions are given in \cref{rmk:extension}.

Let \(R_n\), \(L_n\), and \(p_n\) be total reserves, liquidity, and log marginal price calculated from total reserves \emph{after} the arbitrageurs' trades at block \(n\) (i.e., \(L_n = \varphi(R_n), p_n = \marginp(R_n), R_n = R(L_n, p_n)\)).
Let \(s_n\) be the log of the true market price at block \(n\) which follows Brownian motion as we discussed in \cref{sec:prelim}, so that \(\varepsilon_n \coloneqq (s_n - s_{n - 1}) \sim \mathcal N (0, \sigma^2 \Delta t )\).
We define \(g_n^\mathrm{top} \coloneqq s_n - p_{n-1}\) as the gap between the true market price and the marginal price of active reserves \emph{before} the arbitrageurs' trades at block \(n\).
Throughout the paper, we will often denote it by \(g_n\) when no confusion arises.
We define \(g_n^\mathrm{bot} \coloneqq s_n-p_n \) as the gap between the true price and the marginal price induced from \emph{total} reserves\footnote{
   Thus, in fact \(p_n\) is not the marginal price that uninformed traders will encounter when they come to trade; it will be \( s_n\), the marginal price at active reserves. We elaborate on this in \cref{sec:opt}.
} \emph{after} the arbitrageur's trades as well.
As mentioned earlier, the biggest difference between CFMM and its corresponding PA-AMM is that in PA-AMM, the price gap is not entirely closed within a single block.
The following proposition tells us that there exists a unique stationary distribution \(\pi_{\Delta t} \) of the top-of-the-block gap process \( \{g_n\}_{n \in \mathbb N} \) with the estimation of its second moment up to leading order.

\begin{prop}
  \label{prop:gap_stationary}
  For any \(\Delta t > 0 \) and \(\lambda \in (0,1] \) there exists a unique stationary distribution \(\mathcal G_{\Delta t} \) of the top-of-the-block gap process \( \{g_n\}_{n \in \mathbb N} \), and for \(g \sim \mathcal G_{\Delta t} \), as \(\Delta t \to 0\), we have
  \begin{equation}\label{eq:gap_second_moment}
  \E [g^2] = \frac{\sigma^2 \Delta t}{\lambda (2 - \lambda)} + \bigO (\Delta t^{3/2}) \, .
  \end{equation}
\end{prop}
\begin{proof}
  See \cref{appii:proof_gap_stationary}.
\end{proof}

\begin{remark}\label{rmk:extension}
  The proof of \Cref{prop:gap_stationary} relies on two structural properties of the induced one-step gap map: (i) the gap update can be written as an iterated random function \(g_{n+1} = \Psi(g_n) + \varepsilon_{n+1}\) (or more generally \(g_{n+1} = \Psi (p_{n-1}, g_n) + \varepsilon_{n+1}\)), and (ii) the map \(\Psi\) is uniformly contractive, i.e., \( \sup_g | \Psi'(g) | < 1 \) or \( \sup_{p, \, p+g \in K} | \partial_g \Psi(p,g) | < 1 \) on a compact price region \( K \).
  Therefore, for any CFMM whose marginal-price map \( \marginp\) is sufficiently smooth and non-degenerate on a compact region \( K \) of log prices, the same argument yields existence and uniqueness of a stationary distribution for the gap process as long as the pool state and the reference price remain in \( K \).
  In practice, this “compactness” assumption can be interpreted as restricting attention to horizons over which the asset price remains within a plausible range with high probability; the constants in the moment bounds then depend on \(K\).
\end{remark}

From the above proposition, we derive several consequences that we use throughout the paper.
All of them follow directly from (i) the explicit one-block update map and (ii) a second-order Taylor expansion in the gap.
A generalization to typical CFMMs, omitted here, contains a coefficient that is dependent on \(p_{n-1}\) and the trading function \(\varphi\) in the \(g_n^2\) term.

\begin{coro}
  \label{cor:linearized_gap_valid}
  Let \(\varphi\) be a G3M with weights \((\theta,1-\theta)\).
  Then, for any fixed \(\delta>0\), there exists \(C_\delta<\infty\) depending only on \(\delta,\theta,\lambda\), such that for all \(|g|\le \delta\),
  \begin{equation}
    g_{n+1} = \varepsilon_{n+1} + (1-\lambda) g_n + r(g_n),
    \qquad |r(g_n)|\le C_\delta g_n^2.
    \label{eq:linearized_g3m_with_remainder}
  \end{equation}
  In particular, \(C_\delta \leq 1/8\) for any \(\delta > 0\).
\end{coro}

Next, we show that liquidity per unit of LP supply grows over time, even with a zero swap fee rate, due to the concavity of \(\varphi\).

\begin{coro}
  \label{cor:liquidity_growth_g3m}
  Let \(\varphi(x,y)=x^\theta y^{1-\theta}\) for \(\theta \in (0,1)\) and \(\ell_n \coloneqq \log L_n\), the \emph{log-liquidity} of pool at block
  \(n\) after arbitrage.
  Then as \(g_n\to 0\),
  \begin{equation}
    \ell_n-\ell_{n-1}
    = \frac12 \lambda(1-\lambda)\,\theta(1-\theta)\, g_n^2 \;+\; O(|g_n|^3).
    \label{eq:liq_growth_g3m}
  \end{equation}
  Consequently, under the stationary regime \(g_n\sim \mathcal G_{\Delta t}\) and \(\Delta t\to 0\),
  \begin{equation}
    \lim_{\Delta t\to 0} \frac{\E[\ell_n-\ell_{n-1}]}{\Delta t}
    = \frac12 \lambda(1-\lambda)\theta(1-\theta)\lim_{\Delta t\to 0}\frac{\E[g_n^2]}{\Delta t}
    = \frac{\theta(1-\theta)\sigma^2}{2}\cdot \frac{1-\lambda}{2-\lambda}.
    \label{eq:inst_liq_growth_g3m}
  \end{equation}
\end{coro}

Notice that the coefficient \(\lambda(1-\lambda)\) in \Cref{eq:liq_growth_g3m} is maximized at \(\lambda=\tfrac12\), i.e., for a fixed small gap \(g_n\), this choice maximizes the leading-order \emph{per-block} increase of log-liquidity.\footnote{
  This is what FM-AMM, as introduced in \cite{canidio2023arbitrageurs}, maximizes.
  The main difference is that FM-AMM closes the gap entirely, rather than partially, by leveraging the presence of a batcher that orchestrates the entire order flow from various entities.
}
In contrast, the instantaneous rate is decreasing in \(\lambda\) on \((0,1]\), and is maximized in the limit \(\lambda \to 0\).

The last is an asymptotic analysis of loss-versus-rebalancing (LVR) for PA-AMM.

\begin{coro}
  \label{cor:lvr_g3m}
  Let \(W_n\) be the pool wealth at block \(n\) after arbitrage, with the reference price as the pool's marginal price, \begin{equation}
  W_n \coloneqq W(L_n, P_n)= e^{p_n} x(L_n, p_n) + y(L_n, p_n).
  \end{equation}
  Then as \(g_n\to 0\), the loss-versus-rebalancing at block \(n\), \begin{equation}
  \LVR_n \coloneqq W(L_{n-1},P_{n-1}, S_n) - W(L_n,P_n, S_n),
  \end{equation} satisfies \begin{equation}
    \LVR_n
    = \frac{\lambda}{2}\,\theta(1-\theta)\,W_{n-1}\, g_n^2 \;+\; O(|g_n|^3).
    \label{eq:lvr_g3m}
  \end{equation}
  In particular, the \emph{normalized} loss-versus-rebalancing, \(\LVR^{\mathrm N}_n \coloneqq \LVR_n /  W(L_{n-1},P_{n-1}, S_n)\), satisfies \begin{equation}
    \LVR^{\mathrm N}_n = \frac{\lambda}{2}\,\theta(1-\theta)\, g_n^2 \;+\; O(|g_n|^3).
    \label{eq:norm_lvr_g3m}
  \end{equation}
  Under the stationary regime \(g_n\sim\mathcal G_{\Delta t}\), \begin{equation}
    \lim_{\Delta t\to 0} \frac{\E[\LVR_n/W_{n-1}]}{\Delta t}
    = \frac{\lambda}{2}\theta(1-\theta)\lim_{\Delta t\to 0}\frac{\E[g_n^2]}{\Delta t}
    = \frac{\theta(1-\theta)\sigma^2}{2(2-\lambda)}.
    \label{eq:inst_norm_lvr_rate_g3m}
  \end{equation}
\end{coro}

By plugging \(\theta = \frac{1}{2}\) and \(\lambda = 1\) into \Cref{cor:lvr_g3m}, one recovers the result in \cite{milionis2022automated} on the instantaneous rate of LVR for CPMMs.

We note that \Cref{prop:gap_stationary}, \Cref{cor:liquidity_growth_g3m}, and \Cref{cor:lvr_g3m} together reveal a fundamental trade-off between the price gap and LVR reduction.

In short, as \(\lambda\) decreases, liquidity per unit of LP token supply may grow over time even without swap fees, and LVR decreases.
However, a smaller \(\lambda\) also increases price discrepancy and tracking error, and reduces liquidity per block.
In a more realistic model that allows positive swap fee rates and accounts for order flow from noise traders, the two will have mixed effects.
Reduced \(\lambda\) allows the LPs to set a lower swap fee rate and attract more retail order flow while incurring the same adverse-selection cost as their CFMM equivalent.
However, as liquidity per block declines, price impact as a function of swap size will steepen and eventually overwhelm the benefit of the lower fee rate.
This naturally raises the question of how to define and select the ``optimal'' \(\lambda\), which we discuss in subsequent sections with two different goals.

%% file: optimization.tex
\section{Optimal \texorpdfstring{$\lambda$}{Lambda}}\label{sec:opt}

As shown in \cref{sec:pa_amm}, PA-AMM introduces a fundamental trade-off: smaller activeness \(\lambda\) reduces the value leakage to arbitrageurs, but increases the typical price gap between the pool's marginal price and the true price.
For G3Ms, this admits a natural interpretation as a trade-off between transaction costs (\(\TC\)) and the tracking error (\(\TE\)) of a constant-weight portfolio that aims to keep each asset's capital allocation weight close to its target \citetext{\citealp{evans2021optimal}; \citealp*{angeris2023geometry}}

More specifically, liquidity providers (LPs), as investors of such portfolios, face two sources of dissatisfaction: (i) deviations of the pool's post-arbitrage portfolio weights from their target, i.e., the tracking error (\(\TE\)),\footnote{Note that the arbitrageurs always come first at each block, thus what investors (LPs) face is the post-arbitrage portfolio.} and (ii) the cost of maintaining the target weights (\(\TC\)).
The objective functional of investors (LPs) can thus be expressed as \( \TE + \gamma \TC \), where \(\gamma > 0\) weights the investor's sensitivity to transaction costs relative to tracking error.\footnote{
    A prominent example in DeFi with such an objective is a perpetual demand lending pool (PDLP), such as Jupiter's JLP pool or GMX's GLP pool.
    A PDLP holds a multi-asset inventory that traders borrow to support positions on an associated perpetual exchange, and its target-weight mechanism incentivizes reserve updates that move the current price-weighted composition toward a prescribed target \citep{chitra2025perpetual}.
    The target composition controls the pool's diversification and the inventory available to meet borrowing demand.
    Although deployed PDLPs generally use oracle-based prices and dynamic fees rather than a G3M invariant, their inventory-control problem has the same economic structure: faster rebalancing reduces target-weight deviations, but exposes a larger share of the pool's inventory to the trades that restore those weights.
} 
The optimal level of activeness, denoted as \(\lambda^\ast\), is then defined as the minimizer of this functional.

In this section, we formalize this trade-off as an infinite-horizon stochastic control problem, where \(\{\lambda_n\}_{n\in\mathbb N}\) is the control (which could be either dependent on or independent of the state), and provide a closed-form solution that is \emph{near optimal}.
To elaborate, we first solve the approximate version of the original problem, which is a linear-quadratic-Gaussian control problem that admits a closed-form state-independent optimal policy \( \widetilde \lambda \). 
We then demonstrate that this policy remains a sufficiently good policy for the original problem, in the sense that the objective under this policy \( J(g;\widetilde \lambda) \) is of the \( (1 + \smallo(1)) \) scale of the value function \( V(g) \coloneqq \inf_{\lambda} J(g;\lambda) \) under the regime that \( g = \bigO(\sqrt{\Delta t}) \), for small enough \( \Delta t \).

\subsection{The Problem}

We start by defining two notions introduced and used loosely earlier, the tracking error \(\TE\) and the transaction cost \(\TC\), which will together form a one-stage objective.
Recall that a G3M with weights \((\theta,1-\theta)\) can be interpreted as a constant-weight portfolio; when the pool's marginal price coincides with the true price, the capital allocation to risky asset \(X\) equals \(\theta\).
Accordingly, we define the post-arbitrage risky-asset weight at block \(n\) by
\begin{equation}\label{eq:optiii_post_arbitrage_weight}
    w_n\coloneqq\frac{S_nx_n}{S_nx_n+y_n} \, ,
\end{equation} where \((x_n,y_n)=R_n\) denotes the total reserves after arbitrage at block \(n\), and \(S_n=e^{s_n}\) is the true price.
The tracking error at block \(n\) is then defined as the squared deviation of this weight from its target, \begin{equation}\label{eq:optiii_tracking_error_block}
    \TE_n\coloneqq(w_n-\theta)^2 \, .
\end{equation}
Moreover, for a top-of-the-block gap \(g\) and activeness \(\lambda\), since the exact post-arbitrage gap map of a G3M is \begin{equation}\label{eq:optiii_exact_gap_map}
    \Psi_\lambda(g)\coloneqq g-\log\!\left(1-\lambda+\lambda e^{\theta g}\right)+\log\!\left(1-\lambda+\lambda e^{-(1-\theta)g}\right) \, ,
\end{equation} the \(\TE\) as a function of \(g\) and \(\lambda\) is \begin{equation}\label{eq:optiii_exact_te}
    \TE(g,\lambda) \coloneqq \left(\frac{\theta e^{\Psi_\lambda(g)}}{\theta e^{\Psi_\lambda(g)}+1-\theta}-\theta\right)^2 \,,
\end{equation} and \(\TE_n = \TE (g_n , \lambda_n)\).

For \(\TC\), we use \(\LVR\); by definition, it indicates the loss to arbitrageurs, so it is natural to use it as a metric.
Since \(\TE \in [0, 1]\) is a ratio, we also normalize \(\TC\) by the value of the pre-arbitrage reserves marked at the current true price as well, i.e., \(\TC \coloneqq \LVR^{\mathrm N}\), which was introduced in \cref{cor:lvr_g3m}.
Similarly, expressing \(\TC\) as a function of \(g\) and \(\lambda\), one gets: \begin{equation}\label{eq:optiii_exact_normalized_lvr}
\TC(g,\lambda) 
\coloneqq \LVR^{\mathrm N}
= \lambda\left(1-\frac{e^{\theta g}}{\theta e^g+1-\theta}\right) \,,
\end{equation} with \(\TC_n = \TC(g_n, \lambda_n)\).

Let \(\gamma>0\) denote the relative weight placed on \(\TC\).
Normalizing with \( \theta^2(1-\theta)^2\) for notational convenience, the one-stage objective is then defined as \begin{equation}
    \ell(g, \lambda) \coloneqq \frac{\TE (g, \lambda) + \gamma \TC(g, \lambda)}{\theta^2(1-\theta)^2}.
\end{equation}
Now we define the problem.
Fix a discount rate \(\rho>0\), and with a slight abuse of notation, let \(\Lambda\) denote the set of admissible policies \(\lambda=\{\lambda_n\}_{n\geq0}\), where \(\lambda_n\) stands for both the function from \(\R\) to \([\underline\lambda,1]\) and also the image of that function, i.e., \(\lambda_n = \lambda_n(g_n)\).
The optimization problem is: \begin{alignat}{2} \label{eq:optiii_exact_problem}
\min_{{\lambda=\{\lambda_n\}_{n\geq0}\in\Lambda}} &\quad
& J(g;\lambda) &\coloneqq \E\!\left[\sum_{n=0}^{\infty}e^{-\rho n\Delta t}\ell(g_n,\lambda_n)\,\bigg|\,g_0=g\right] \\
\text{subject to} &
& g_{n+1} &= \Psi_{\lambda_n}(g_n)+\varepsilon_{n+1}, \\
&
& \varepsilon_{n+1} &\sim \mathcal N(0,\sigma^2\Delta t) \quad \text{i.i.d.,} \\
&
& \lambda_n &= \lambda_n(g_n)\in[\underline\lambda,1] \, .
\end{alignat}
We denote the value function of this problem by \begin{equation} \label{eq:optiii_exact_value_function}
V(g)\coloneqq\inf_{\lambda\in\Lambda}J(g;\lambda) \, .
\end{equation}

\subsection{The Auxiliary Problem and Its Solution}

Although the original problem's one-stage objective and the state's evolution are both nonlinear, first-order approximations significantly simplify the problem.
In this section, we solve this auxiliary problem in the regime where the gap remains small.
The following lemma provides the leading-order approximations and the resulting uniform residual bound:
\begin{lem}\label{lem:optiii_loss_expansion}
Uniformly over \(\lambda\in[\underline\lambda,1]\), as \(g\to0\), \begin{align}\label{eq:optiii_loss_expansions}
    \TE(g,\lambda) &= \theta^2(1-\theta)^2(1-\lambda)^2g^2+\bigO(|g|^3), \\
    \TC(g,\lambda) &= \frac{\lambda}{2}\theta(1-\theta)g^2+\bigO(|g|^3) \, .
\end{align}
Consequently, for the leading-order approximation of the one-stage objective \begin{equation}\label{eq:optiii_approximated_stage_loss}
\widetilde\ell(g,\lambda) \coloneqq g^2 \left((1 - \lambda)^2 + \frac{\gamma \lambda}{2\theta (1 - \theta)} \right) \, ,
\end{equation} there exists a constant \(C<\infty\), independent of \(g\) and \(\lambda\), such that \begin{equation}\label{eq:optiii_global_stage_residual}
|\ell(g,\lambda)-\widetilde\ell(g,\lambda)|\leq C |g|^3, \qquad \forall g\in\R, \quad \forall\lambda\in[\underline\lambda,1] \, .
\end{equation}
\end{lem}
\begin{proof}
The approximations are just Taylor expansions.
The cubic order of error is a direct result of the boundedness of \( 0 \le \TE \leq 1 \) and \(0 \leq \TC \leq 1\).
\end{proof}

Similarly, the approximated gap map is \begin{equation}\label{eq:optiii_approximated_gap_map}
\widetilde\Psi_\lambda(g)\coloneqq(1-\lambda)g \, ,
\end{equation} and the auxiliary problem becomes: \begin{alignat}{2}
\label{eq:optiii_approximated_problem}
\min_{\lambda=\{\lambda_n\}_{n\geq0}\in\Lambda} &\quad
& \widetilde J(g;\lambda) &\coloneqq \E\!\left[\sum_{n=0}^{\infty}e^{-\rho n\Delta t}\widetilde\ell(\widetilde g_n,\lambda_n)\,\bigg|\,\widetilde g_0=g\right] \\
\text{subject to} &
& \widetilde g_{n+1} &= (1-\lambda_n)\widetilde g_n+\varepsilon_{n+1}, \\
&
& \varepsilon_{n+1} &\sim \mathcal N(0,\sigma^2\Delta t) \quad \text{i.i.d.,} \\
&
& \lambda_n &= \lambda_n(\widetilde g_n)\in[\underline\lambda,1] \, .
\end{alignat}
The value function of this problem will be denoted as
\begin{equation}\label{eq:optiii_approximated_value_function}
\widetilde V(g)\coloneqq\inf_{\lambda\in\Lambda}\widetilde J(g;\lambda) \, .
\end{equation}
The auxiliary problem is now linear-quadratic-Gaussian and is exactly solvable. 
For notational convenience, let \(
\beta\coloneqq e^{-\rho\Delta t} ,
\) and for \(\gamma' \coloneqq \frac{\gamma}{2\theta(1-\theta)} > 0 \), define the constants
\begin{align}\label{eq:optiii_limiting_coefficients}
\widetilde\lambda_0 &= \frac{1+\sqrt{1+2\gamma'}}{1+\gamma'+\sqrt{1+2\gamma'}}, \\
\widetilde v &= \frac{\gamma'-1+\sqrt{1+2\gamma'}}{2} \, .
\end{align}
The following theorem then fully characterizes the solution of \cref{eq:optiii_approximated_problem}:
\begin{thm}[Solution of the auxiliary problem]\label{thm:optiii_approximated_solution}
Suppose \(\gamma'>0\) and \(0<\underline\lambda<\widetilde\lambda_0\).
For all sufficiently small \(\Delta t\), the auxiliary problem has a state-independent optimal policy given by
\begin{equation}\label{eq:optiii_tilde_lambda}
\widetilde\lambda=1-\frac{\gamma'}{1+\beta\gamma'+\sqrt{(1+\beta\gamma')^2-\beta\gamma'^2}} \, .
\end{equation}
The value function \(\widetilde V(g)\) is:
\begin{equation}
\widetilde V(g) = \widetilde v_2g^2+\widetilde v_0 \, ,
\end{equation} where \begin{equation}\label{eq:optiii_tilde_value}
\widetilde v_2 = \frac{\beta\gamma'-1+\sqrt{(1+\beta\gamma')^2-\beta\gamma'^2}}{2\beta}, \quad \text{and} \quad
\widetilde v_0 = \frac{\beta}{1-\beta}\widetilde v_2\sigma^2\Delta t \, .
\end{equation}
Moreover, as \(\Delta t\to0\),
\begin{equation}\label{eq:optiii_tilde_limits}
\widetilde\lambda=\widetilde\lambda_0+\bigO(\Delta t), \qquad \widetilde v_2=\widetilde v+\bigO(\Delta t), \qquad \widetilde v_0=\frac{\sigma^2}{\rho}\widetilde v+\bigO(\Delta t) \, .
\end{equation}
\end{thm}
\begin{proof}
See \cref{appii:proof_optiii_approximated_solution}.
\end{proof}

In particular, as shown in \cref{fig:te_tc_lambda}, the optimal activeness \(\widetilde\lambda_0\) is strictly decreasing in \(\gamma'\), so a larger weight on \(\TC\) leads to slower rebalancing.
Conversely, letting \(\gamma' \to 0\), i.e., prioritizing the minimization of the tracking error, \(\widetilde\lambda_0 \to 1\) and recovers the typical CFMM.

\begin{figure}[h]
    \centering
    \includegraphics[width=\textwidth]{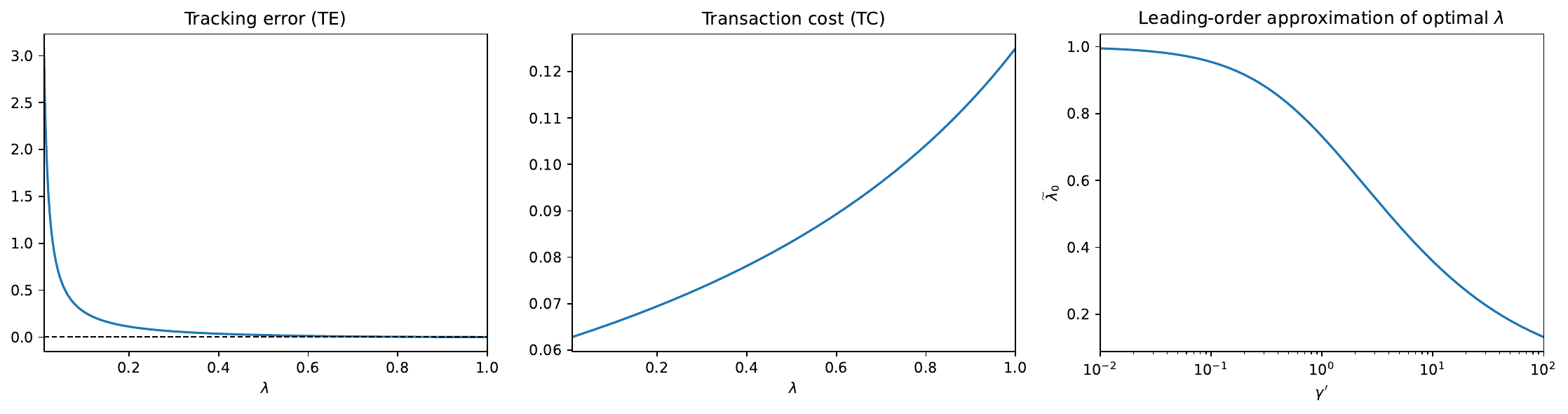}
    \caption{
        \textbf{Left.} The mean tracking error \(\TE\) for a given \( \lambda \).
        \textbf{Middle.} The mean transaction cost \(\TC\) for a given \( \lambda \).
        \textbf{Right.} The optimal activeness \(\widetilde\lambda_0\) as a function of \(\gamma'\).
    }
    \label{fig:te_tc_lambda}
\end{figure}

\subsection{Near Optimality of \texorpdfstring{$\widetilde\lambda$}{the Optimal Policy of Auxiliary Problem}}

We now evaluate the policy \(\widetilde\lambda\), which was optimal in the auxiliary problem, against the exact objective and the nonlinear gap dynamics.
The following theorem tells us that \(\widetilde\lambda\) is still \emph{near-optimal}:

\begin{thm}[Near-optimality]\label{thm:optiii_near_optimality}
Suppose the conditions of \cref{thm:optiii_approximated_solution} hold and \(\sigma>0\).
There exist \(C<\infty\) and \(\Delta t_0>0\), independent of \(g\) and \(\Delta t\in(0,\Delta t_0]\), such that,  for any fixed \(M<\infty\), if \(|g|\leq M\sqrt{\Delta t}\) then
\begin{equation}\label{eq:optiii_near_optimality_global}
0\leq J(g;\widetilde\lambda)-V(g)\leq C\left(|g|^3+\sqrt{\Delta t}\right), \qquad 0<\Delta t\leq\Delta t_0 \, .
\end{equation}
Consequently, uniformly over \(|g|\leq M\sqrt{\Delta t}\),
\begin{equation}\label{eq:optiii_near_optimality_local}
J(g;\widetilde\lambda)\leq V(g)+\bigO(\sqrt{\Delta t}) \, .
\end{equation}
More precisely, uniformly over the same local regime,
\begin{equation}\label{eq:optiii_value_scaling}
V(g) = \frac{\sigma^2}{\rho}\widetilde v+\bigO(\sqrt{\Delta t}) \quad \text{and} \quad
J(g;\widetilde\lambda) = \frac{\sigma^2}{\rho}\widetilde v+\bigO(\sqrt{\Delta t}) \, .
\end{equation}
Hence both \(V(g)\) and \(J(g;\widetilde\lambda)\) are of \(\Theta(1)\), and
\begin{equation}\label{eq:optiii_relative_near_optimality}
J(g;\widetilde\lambda)=V(g)\left(1+\bigO(\sqrt{\Delta t})\right)=V(g)(1+\smallo(1)) \, .
\end{equation}
\end{thm}
\begin{proof}
See \cref{appii:proof_optiii_near_optimality}.
\end{proof}

%% file: experiments.tex
\section{Numerical Experiments}\label{sec:exp}

In \cref{sec:opt}, we introduced the objective from the perspective that LPs are investors seeking exposure to a constant-weight portfolio; deviations of realized weights from their targets are a cost in themselves.
If we see LPs as delta-neutral market makers, who offset the inventory acquired from each trade and are therefore indifferent to the composition of the reserves, such an objective is inappropriate and should be replaced with a new one.
For them, the relevant performance measure is the markout PnL of the trades the pool executes.
In this section, we therefore evaluate PA-AMMs by their cumulative markout PnL over various configurations of activeness \(\lambda\) and fee rates \(f\) via Monte Carlo simulation, with both arbitrageurs and noise traders present.

In the experiment, a PA-AMM with parameters \((\lambda, f)\) competes for the order flow of noise traders with a CFMM pool, which we refer to as the \emph{normalizer} pool.
At the top of each block, arbitrageurs trade each pool up to the boundary of its no-arbitrage band.
Noise traders' orders do not depend on prices; only their split between the two pools does, and that split maximizes each order's output.
Since the PA-AMM with \(\lambda = 1\) is the CFMM with the same fee rate, the comparison across \(\lambda\) at a fixed \(f\) isolates the effect of activeness, and the joint maximizer over \((\lambda, f)\) shows how the choice of the fee rate interacts with it.
Similar frameworks have been used in \cite{campbell2025optimal}, \cite{volosnikov2024impact} and in a public competition.\footnote{
  See \url{https://www.paradigm.xyz/puzzles/prop-amm}.
  Our simulator follows the execution and accounting rules of the competition's reference implementation.
}

\subsection{Setup}

\paragraph{Performance Measure.}

Consider a netted trade at block \(j\) executed against the PA-AMM, and let \((\Delta x_j, \Delta y_j)\) be the resulting change in its reserves and \(S_j\) the external price of \(X\) at the time of execution.
The markout of the trade is the change in the value of the reserves marked at the external price, \begin{equation}\label{eq:exp_markout}
  \PnL_j \coloneqq S_j \Delta x_j + \Delta y_j \, ,
\end{equation} which is positive when the pool sells \(X\) above \(S_j\) or buys it below.
It is the PnL of a delta-neutral LP that offsets each trade's inventory change at the external price, assuming no hedging costs.
The cumulative markout up to block \(T\) is \(\PnL_T \coloneqq \sum_{j \leq T} \PnL_j\).

Trades of arbitrageurs have negative markouts by construction, and their sum is the loss to arbitrageurs net of the swap fees they pay; when \(f = 0\), it equals \(-\sum_{n \leq T} \LVR_n\) with \(\LVR_n\) as in \cref{cor:lvr_g3m}.
The markout of a noise trader's trade is the swap fee and the price impact she pays, less the deviation of the pool's marginal price from \(S_j\), and is therefore typically positive.
Accordingly, we decompose \(\PnL_T = \PnL_T^{\mathrm{arb}} + \PnL_T^{\mathrm{noise}}\) by the counterparty and report both components.

\paragraph{Sample Generation.}

The initial external price is fixed as \(S_0 = 100\).
Both pools use \(\varphi(x,y) = \sqrt{xy}\) as the invariant function, with initial reserves \((x_0, y_0) = (100, 10{,}000)\) for PA-AMM, so that initial wealth is \(20{,}000\).
A \emph{path} consists of \(10{,}000\) blocks, over which the external price follows a driftless geometric Brownian motion with volatility per block \(\sigma\), i.e., \[
    S_{t} = S_{t-1} \exp{\left( -\frac{\sigma^2}{2} + \sigma Z_t \right)}, \quad Z_t \sim \mathcal N(0, 1) \, .
\]
In each block, the number of noise orders is Poisson with parameter \(\nu\),\footnote{
The probability mass function (PMF) of the Poisson distribution with parameter \(\nu\) is \( \mathbb P (N = k) = \frac{e^{-\nu} \nu^k}{k!} \) for nonnegative integer \(k\), and its mean and variance are both \(\nu\).
} and each order is a buy or a sell of \(X\) with equal probability, with a lognormally distributed size with mean \(m\) in units of the quote asset \(Y\).
The parameters \(\sigma\), \(\nu\), \(m\), the fee rate, and the initial reserves of the normalizer pool are drawn independently at the start of each path from the uniform distributions in \cref{tab:exp_market}, and are held fixed along the path.
Thus, each sample consists of randomly picked parameters, the price path, and the noise order flow generated from those parameters.
Randomizing them, rather than fixing a single market, is what makes the comparison across \((\lambda, f)\) informative about a range of market conditions rather than about one calibration.

We generate \(10{,}000\) samples and evaluate every pair in the grid \(\lambda \in \{0.05, 0.10, \ldots, 1\}\) and \(f \in \{5, 10, \ldots, 100\}\) bp, or \(400\) coordinates, on those fixed samples.
Common samples make the comparison paired and sample-by-sample, so the standard error of the difference between two values of \(\lambda\) at the same fee rate is much smaller than the standard error at either level.
Reported standard errors of levels are across samples, and those of differences are computed from the path-by-path differences \(\PnL_T(\lambda, f) - \PnL_T(1, f)\).

\begin{table}[htbp]
  \centering
  \small
  \begin{tabular}{lrrrrrr}
    \hline
                                                & Support & Mean & SD & p5 & p50 & p95 \\
    \hline
    \multicolumn{7}{l}{\emph{Panel A. Parameters drawn per path}} \\
    Volatility per block \(\sigma\) (bp)        & \([1, 70]\)       & 35.5 & 20.0 & 4.3 & 35.4 & 66.7 \\
    Mean number of orders per block \(\nu\)     & \([0.4, 1.2]\)    & 0.80 & 0.23 & 0.44 & 0.80 & 1.16 \\
    Mean order size \(m\)                       & \([12, 28]\)      & 20.02 & 4.61 & 12.82 & 20.00 & 27.21 \\
    Fee rate of the normalizer pool (bp)        & \(\{30, \ldots, 80\}\) & 54.9 & 14.7 & 32.0 & 55.0 & 78.0 \\
    Size of the normalizer pool (\(\times\) PA-AMM) & \([0.4, 2.0]\) & 1.19 & 0.46 & 0.48 & 1.19 & 1.92 \\
    \hline
    \multicolumn{7}{l}{\emph{Panel B. Realized per path}} \\
    Number of noise orders                      & & 7{,}993 & 2{,}307 & 4{,}400 & 7{,}964 & 11{,}589 \\
    Noise notional (\(\times 10^{3}\))          & & 160.1 & 60.3 & 74.7 & 151.6 & 272.2 \\
    Median order size                           & & 9.75 & 2.25 & 6.24 & 9.73 & 13.25 \\
    \(95\)th percentile of order size           & & 70.1 & 16.3 & 44.8 & 70.0 & 95.6 \\
    Terminal price ratio \(S_T / S_0\)          & & 1.00 & 0.45 & 0.42 & 0.96 & 1.76 \\
    \hline
  \end{tabular}
  \caption{
    Parameters used in the simulation and descriptive statistics of the samples.
  }
  \label{tab:exp_market}
\end{table}

\subsection{Results}\label{sec:exp_results}

\Cref{fig:exp_grid} reports the mean cumulative markout over the grid, and \cref{tab:exp_results} reports the optimal activeness at each fee rate together with the paired gain over the fully active pool.
At every fee rate in the grid, the mean cumulative markout is maximized at an interior \(\lambda\).
The paired gain over \(\lambda = 1\) at the same fee rate is largest where the fee protects LPs least, at \(388.3\) units of \(Y\) for \(f = 5\) bp, and declines to \(10.5\) at \(f = 100\) bp.

\begin{figure}[htbp]
  \centering
  \includegraphics[width=\textwidth]{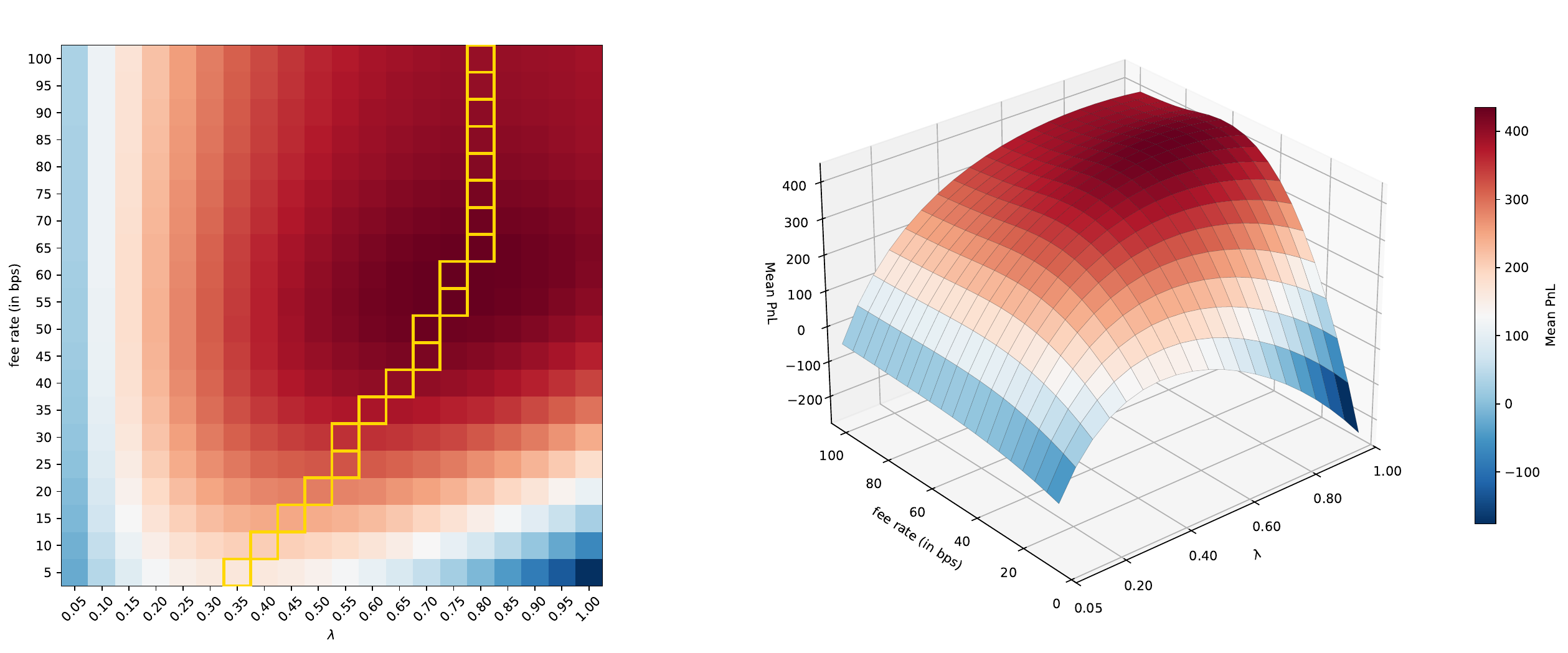}
  \caption{
    \textbf{Left.} The heatmap of mean cumulative markout \(\PnL_T\), in units of the quote asset \(Y\).
     The cell outlined in yellow corresponds to the row-wise optimal combination of parameters.
    \textbf{Right.} The same heatmap plotted as surface in three dimensions; the joint maximizer is at \((\lambda^\ast, f^\ast) = (0.75, 60 \text{ bp})\).
  }
  \label{fig:exp_grid}
\end{figure}

The gains are not confined to a narrow neighborhood of \(\lambda^\ast\): at \(f = 30\) bp, every \(\lambda\) between \(0.25\) and \(0.95\) beats the fully active pool, and the values within one standard error of the maximum span \(\lambda \in [0.55, 0.60]\).
Moreover, the observed difference in performance is starker at low fee rates: a fully active pool scores negative markout on average at \(f \leq 15\) bp, with a mean of \(-263.7\) at \(f = 5\) bp, while a PA-AMM configured with optimal activeness at each fee rate remains profitable and also maximizes the markout difference by \(388.3\).

\begin{table}[htbp]
  \centering
  \small
  \begin{tabular}{rrrrr}
    \hline
    \(f\) (bp) & \(\lambda^\ast\) & \(\PnL_T\) at \(\lambda^\ast\) & \(\PnL_T\) at \(\lambda = 1\) & Paired gain \\
    \hline
    5 & 0.35 & 124.6 (3.4) & \(-263.7\) (3.4) & 388.3 (2.46) \\
    10 & 0.40 & 173.5 (3.4) & \(-139.6\) (3.2) & 313.1 (2.22) \\
    20 & 0.50 & 264.4 (3.4) & 63.5 (2.9) & 200.8 (1.77) \\
    30 & 0.55 & 341.1 (3.5) & 217.2 (2.9) & 123.8 (1.48) \\
    40 & 0.65 & 397.0 (3.6) & 322.5 (3.0) & 74.5 (1.10) \\
    50 & 0.70 & 428.0 (3.7) & 384.5 (3.3) & 43.5 (0.88) \\
    60 & 0.75 & 436.1 (3.8) & 410.2 (3.5) & 25.9 (0.69) \\
    75 & 0.80 & 417.1 (3.8) & 401.3 (3.5) & 15.8 (0.49) \\
    100 & 0.80 & 391.2 (3.8) & 380.7 (3.5) & 10.5 (0.38) \\
    \hline
  \end{tabular}
  \caption{
    Optimal activeness by the fee rate and the average paired gain against \(\lambda = 1\).
    \(\lambda^\ast\) maximizes the mean cumulative markout \(\PnL_T\) at the given fee rate, and the last column is the mean of the path-by-path difference against \(\lambda = 1\).
    Standard errors across the \(10{,}000\) paths are in parentheses.
    Markouts are in units of the quote asset \(Y\).
  }
  \label{tab:exp_results}
\end{table}

\Cref{fig:exp_profiles} shows where the gain comes from.
Lowering \(\lambda\) from \(1\) to \(\lambda^\ast = 0.55\) at \(f = 30\) bp shrinks the loss to arbitrageurs from \(529.1\) to \(314.9\) units of \(Y\), a reduction of \(40\%\), while the revenue from noise traders falls from \(746.4\) to \(656.0\), a reduction of \(12\%\).
It is visible in the routed flow, of which the PA-AMM executes \(63.4\%\) at \(\lambda = 1\) and \(48.4\%\) at \(\lambda^\ast\), so the pool gives up about a quarter of its share of the flow to cut its adverse-selection cost by two-fifths.
Below \(\lambda \approx 0.25\), the ordering reverses, and at \(\lambda = 0.05\) the pool executes only \(7.6\%\) of the flow and earns a negative markout, in line with the order sizes reported in \cref{tab:exp_market}.

\begin{figure}[htbp]
  \centering
  \includegraphics[width=0.5\textwidth]{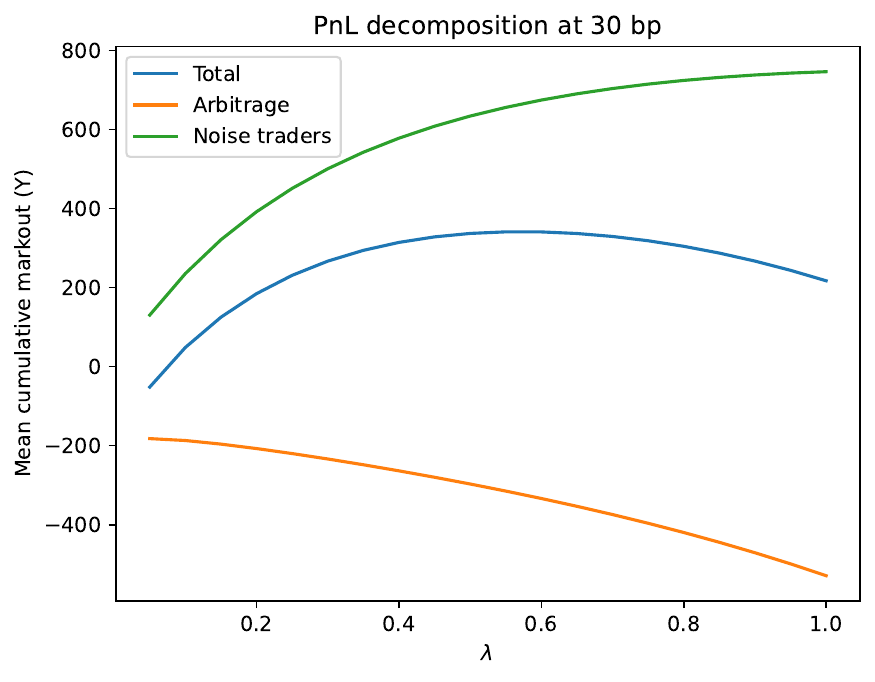}
  \caption{
    The decomposition of PnL at \(f = 30\) bp into the arbitrage component \(\PnL_T^{\mathrm{arb}}\) and the noise-trader component \(\PnL_T^{\mathrm{noise}}\).
  }
  \label{fig:exp_profiles}
\end{figure}

\subsection{Discussion}\label{sec:exp_discussion}
The results show that the partial activeness affects performance by reducing both adverse selection costs and profitable order flow.
Lower activeness leaves the band unchanged and instead reduces the liquidity arbitrageurs can trade against when the external price leaves it, so the loss per arbitrage trade is smaller.
However, as a trade-off, trading at PA-AMM incurs steeper price impact, reducing the share of order flow from noise traders.
This is the mixed effect anticipated at the end of \cref{sec:pa_amm}, and it adds an activeness as a second instrument to the fee-versus-volume trade-off studied in \cite{evans2021optimal}, \cite{campbell2025optimal}, and \cite{baggiani2025optimal}.

The experiment also shows how a delta-neutral LP prices the tracking error.
While such an LP bears no direct cost from the deviation of the weights from their target, a persistent gap between the pool's marginal price and the external price means that the pool quotes a worse price on one side of the market and loses the routed order flow on that side, while giving up part of its fee on the other.
The tracking error thus enters the markout through the forgone order flow, and the interior maximizer in \(\lambda\) reflects the same trade-off as the objective in \cref{sec:opt}, with the relative weight of the two terms set by the order flow rather than chosen by the LP.


%% file: conclusion.tex
\section{Conclusion}\label{sec:conc}

We introduced PA-AMM, which partitions a pool's reserves into active and passive components and permits trades only against the active component within each rebalancing interval.
The activeness parameter \(\lambda\) controls how quickly the pool adjusts its inventory to new information.
A smaller \(\lambda\) limits the inventory exposed to arbitrage in each block and reduces \(\LVR\), but also slows price adjustment and increases the deviation of the pool's asset weights from their targets.

For G3Ms, we further formulated this trade-off as a discounted stochastic control problem with tracking error and normalized \(\LVR\), obtained a closed-form state-independent solution to the auxiliary problem, and established its near-optimality for the exact nonlinear problem when block times are short.

The numerical experiment complements this index-tracking objective with a market-making objective based on cumulative hedged PnL, and examines the interaction among activeness, fees, volatility, and uninformed order flow.
The result shows that optimal \(\lambda\) for each fee rate remains less than \(1\) even under such an objective.

The present study is subject to several limitations. 
From a theoretical perspective, the analysis is conducted under highly idealized assumptions, including zero gas fees and zero swap fee rates.
In the experiment, noise traders' aggregate demand is exogenous, so we do not model how traders, LPs, and competing venues respond to a PA-AMM.
The hedging is also implicitly frictionless, and we assume arbitrageurs bear no additional costs, such as gas costs and taker fees paid at CEXs; accounting for those could materially change the result.
Finally, the external price follows geometric Brownian motion, which reproduces neither jumps nor volatility clustering, which are frequently observed in real data.

In conclusion, our study shows that activeness is a design margin distinct from the swap fee, that its preferred value depends on market conditions, and therefore that the appropriate choice of \(\lambda\) depends on the archetype of the liquidity provider (LP).


%% file: appendix.tex
\appendix
\section{Appendix}

\subsection{Proof of \texorpdfstring{\Cref{prop:gap_stationary}}{Proposition 2}}\label{appii:proof_gap_stationary}
\begin{proof}
  Recall that the underlying CFMM invariant is \(\varphi(x,y)=x^\theta y^{1-\theta}\) with \(\theta\in(0,1)\).
  Fix \(\lambda\in(0,1]\) and let \(\varepsilon_n\sim\mathcal N(0,\sigma^2\Delta t)\) independently and identically distributed.
  For G3M, one checks that \begin{equation}
    \marginp \big( (1-\lambda) R(1,p)+\lambda R(1,p+g)\big) = p + \log\!\big(1-\lambda+\lambda e^{\theta g}\big) - \log\!\big(1-\lambda+\lambda e^{-(1-\theta)g}\big).
  \end{equation}
  Hence the exact top-of-the-block gap recursion is \begin{equation}\label{eq:appii_g3m_recursion}
    g_{n+1}=\varepsilon_{n+1}+\Psi_\lambda(g_n), \qquad \Psi_\lambda(g)\coloneqq g-\log\!\left(1-\lambda+\lambda e^{\theta g}\right)+\log\!\left(1-\lambda+\lambda e^{-(1-\theta)g}\right) \, .
  \end{equation}
  In particular, \(\Psi_\lambda\) does not depend on the price level \(p\).
  Differentiating gives
  \begin{equation}\label{eq:appii_psi_prime}
    \Psi_\lambda'(g)=1-\frac{\lambda\theta e^{\theta g}}{1-\lambda+\lambda e^{\theta g}}-\frac{\lambda(1-\theta)e^{-(1-\theta)g}}{1-\lambda+\lambda e^{-(1-\theta)g}} \, .
  \end{equation}
  Each fraction in \eqref{eq:appii_psi_prime} lies in \([0,\theta]\) and \([0,1-\theta]\), respectively.
  Moreover, the first fraction is at least \(\lambda\theta\) when \(g\geq0\), and the second is at least \(\lambda(1-\theta)\) when \(g\leq0\).
  Therefore \begin{equation}\label{eq:appii_psi_derivative_bound}
    0\leq\Psi_\lambda'(g)\leq1-\lambda\min\{\theta,1-\theta\}\eqqcolon\kappa<1 \, .
  \end{equation}
  Since \(\Psi_\lambda(0)=0\), the mean value theorem yields \begin{equation}\label{eq:appii_global_contraction}
    |\Psi_\lambda(g)|\leq\kappa|g|, \qquad |\Psi_\lambda(g)-\Psi_\lambda(h)|\leq\kappa|g-h|, \qquad \forall g,h\in\R \, .
  \end{equation}
  Thus \(\kappa\) is a global Lipschitz constant of \(\Psi_\lambda\).
  Define the random map \(G_\varepsilon(g)\coloneqq\Psi_\lambda(g)+\varepsilon\).
  The family \(\{G_\varepsilon\}\) is uniformly contractive, so the induced Markov chain admits a unique stationary distribution \(\mathcal G_{\Delta t}\) and converges to it at an exponential rate \cite[Theorem~1.1]{persi1999iterated}.

Let \(g\sim\pi_{\Delta t}\) and let \(\varepsilon\sim\mathcal N(0,\sigma^2\Delta t)\) be independent of \(g\).
Stationarity gives \(g\overset{d}{=}\Psi_\lambda(g)+\varepsilon\).
For any \(p\geq1\), Minkowski's inequality and \eqref{eq:appii_global_contraction} imply
\begin{equation}
\|g\|_p\leq\|\Psi_\lambda(g)\|_p+\|\varepsilon\|_p\leq\kappa\|g\|_p+\|\varepsilon\|_p \, .
\end{equation}
Consequently,
\begin{equation}\label{eq:appii_moments}
\E[|g|^p]\leq\frac{1}{(1-\kappa)^p}\E[|\varepsilon|^p]=\bigO(\Delta t^{p/2}) \, .
\end{equation}

A direct differentiation gives
\begin{equation}
\Psi_\lambda''(g)=-\lambda(1-\lambda)\left[\frac{\theta^2e^{\theta g}}{(1-\lambda+\lambda e^{\theta g})^2}-\frac{(1-\theta)^2e^{-(1-\theta)g}}{(1-\lambda+\lambda e^{-(1-\theta)g})^2}\right] \, .
\end{equation}
For \(\lambda=1\), the map \(\Psi_\lambda\) is identically zero.
For \(\lambda\in(0,1)\), using \(\sup_{z>0}z/(A+Bz)^2=1/(4AB)\), we obtain
\begin{equation}
\sup_{g\in\R}|\Psi_\lambda''(g)|\leq\frac{\theta^2+(1-\theta)^2}{4}\leq\frac14 \, .
\end{equation}
Taylor's theorem, therefore, gives
\begin{equation}\label{eq:appii_psi_expansion}
\Psi_\lambda(g)=(1-\lambda)g+r_\lambda(g), \qquad |r_\lambda(g)|\leq\frac{g^2}{8}, \qquad \forall g\in\R \, .
\end{equation}
Using stationarity and \(\E[\varepsilon]=0\),
\begin{equation}
\E[g^2]=\E[\Psi_\lambda(g)^2]+\sigma^2\Delta t \, .
\end{equation}
Substituting \eqref{eq:appii_psi_expansion} and rearranging gives
\begin{equation}\label{eq:appii_variance_equation}
\lambda(2-\lambda)\E[g^2]=\sigma^2\Delta t+\E[r_\lambda(g)^2]+2(1-\lambda)\E[g\,r_\lambda(g)] \, .
\end{equation}
By \eqref{eq:appii_moments} and \(|r_\lambda(g)|\leq g^2/8\),
\begin{equation}
\E[r_\lambda(g)^2]=\bigO(\Delta t^2), \qquad \E[g\,r_\lambda(g)]=\bigO(\Delta t^{3/2}) \, .
\end{equation}
It follows that
\begin{equation}
\E[g^2]=\frac{\sigma^2\Delta t}{\lambda(2-\lambda)}+\bigO(\Delta t^{3/2}) \, .
\end{equation}
\end{proof}

\subsection{Proof of \texorpdfstring{\Cref{thm:optiii_approximated_solution}}{Theorem}}\label{appii:proof_optiii_approximated_solution}
\begin{proof}
Write \(u\coloneqq1-\lambda\in[0,1-\underline\lambda]\).
The Bellman equation of the auxiliary problem is
\begin{equation}\label{eq:appii_approximated_bellman}
\widetilde V(g)=\min_{u\in[0,1-\underline\lambda]}\left\{(u^2-\gamma' u+\gamma')g^2+\beta\E[\widetilde V(ug+\varepsilon)]\right\} \, .
\end{equation}
Consider a quadratic ansatz
\begin{equation}
\widetilde V(g)=\widetilde v_2g^2+\widetilde v_0 \, .
\end{equation}
Since \(\E[\varepsilon]=0\) and \(\E[\varepsilon^2]=\sigma^2\Delta t\),
\begin{equation}
\E[\widetilde V(ug+\varepsilon)]=\widetilde v_2u^2g^2+\widetilde v_2\sigma^2\Delta t+\widetilde v_0 \, .
\end{equation}
Substitution into \eqref{eq:appii_approximated_bellman} separates the quadratic and constant coefficients as
\begin{align}\label{eq:appii_approximated_coefficients}
\widetilde v_2 &= \min_{u\in[0,1-\underline\lambda]}\left\{(1+\beta\widetilde v_2)u^2-\gamma' u+\gamma' \right\}, \\
\widetilde v_0 &= \frac{\beta}{1-\beta}\widetilde v_2\sigma^2\Delta t \, .
\end{align}
For the unrestricted optimizer, the first-order condition is
\begin{equation}\label{eq:appii_approximated_foc}
2(1+\beta\widetilde v_2)u-\gamma'=0 \, .
\end{equation}
Substituting \eqref{eq:appii_approximated_foc} into the first equation of \eqref{eq:appii_approximated_coefficients} gives \begin{equation}
\widetilde v_2=\gamma'-\frac{\gamma'^2}{4(1+\beta\widetilde v_2)} \, .
\end{equation}
Equivalently, \begin{equation}\label{eq:appii_approximated_riccati}
\beta\widetilde v_2^2+(1-\beta\gamma')\widetilde v_2+\frac{\gamma'^2}{4}-\gamma'=0 \, .
\end{equation}
The (economically) feasible root of \eqref{eq:appii_approximated_riccati} is \begin{equation}
\widetilde v_2=\frac{\beta\gamma'-1+\sqrt{(1+\beta\gamma')^2-\beta\gamma'^2}}{2\beta} \, ,
\end{equation}
thus we get \eqref{eq:optiii_tilde_value}.
Plugging it back to \eqref{eq:appii_approximated_foc}, the optimal \(\widetilde u\) becomes \begin{equation}
\widetilde u\coloneqq1-\widetilde\lambda=\frac{\gamma'}{2(1+\beta\widetilde v_2)}=\frac{\gamma'}{1+\beta\gamma'+\sqrt{(1+\beta\gamma')^2-\beta\gamma'^2}} \, ,
\end{equation}
which proves \eqref{eq:optiii_tilde_lambda}.
As \(\Delta t\to0\), \(\beta\to1\) and \(\widetilde u\) converges to \( 1 - \widetilde\lambda_0\), and the assumption \(\underline\lambda<\widetilde\lambda_0\) therefore ensures that \(\widetilde u\) remains interior for all sufficiently small \(\Delta t\).
The ansatz thus indeed satisfies the Bellman equation, and the bound \(u\leq1-\underline\lambda<1\) makes the discounted terminal term vanish under every admissible policy.
Then the standard verification argument obtained by iterating the Bellman inequality establishes optimality.
Finally, \(\beta=1+\bigO(\Delta t)\) and \begin{equation}
\frac{\beta\Delta t}{1-\beta}=\frac{1}{\rho}+\bigO(\Delta t) \, ,
\end{equation} which proves \eqref{eq:optiii_tilde_limits}.
\end{proof}

\subsection{Proof of \texorpdfstring{\Cref{thm:optiii_near_optimality}}{Theorem}}\label{appii:proof_optiii_near_optimality}
\begin{proof}
We compare Bellman residuals of the original and auxiliary problems.
Throughout the proof, constants are uniform over \(\lambda\in[\underline\lambda,1]\) and sufficiently small \(\Delta t\).

By \cref{lem:optiii_loss_expansion},
\begin{equation}\label{eq:appii_stage_residual}
|\ell(g,\lambda)-\widetilde\ell(g,\lambda)|\leq C|g|^3 \, .
\end{equation}

Write
\begin{equation}
\Psi_\lambda(g)=\widetilde\Psi_\lambda(g)+r_\lambda(g), \qquad \widetilde\Psi_\lambda(g)=(1-\lambda)g \, .
\end{equation}
Since \(\widetilde V(g)=\widetilde v_2g^2+\widetilde v_0\),
\begin{equation}
\widetilde V(\Psi_\lambda(g)+\varepsilon)-\widetilde V(\widetilde\Psi_\lambda(g)+\varepsilon)=\widetilde v_2\left(2(\widetilde\Psi_\lambda(g)+\varepsilon)r_\lambda(g)+r_\lambda(g)^2\right) \, .
\end{equation}
Taking expectations and using \(\E[\varepsilon]=0\) gives
\begin{equation}
\E\!\left[\widetilde V(\Psi_\lambda(g)+\varepsilon)-\widetilde V(\widetilde\Psi_\lambda(g)+\varepsilon)\right]=\widetilde v_2\left(2(1-\lambda)g\,r_\lambda(g)+r_\lambda(g)^2\right) \, .
\end{equation}
The coefficient \(\widetilde v_2\) is uniformly bounded by \cref{thm:optiii_approximated_solution}, and \(|r_\lambda(g)|\leq g^2/8\) by \eqref{eq:appii_psi_expansion}.
Therefore
\begin{equation}\label{eq:appii_continuation_residual}
\left|\E\!\left[\widetilde V(\Psi_\lambda(g)+\varepsilon)-\widetilde V(\widetilde\Psi_\lambda(g)+\varepsilon)\right]\right|\leq C |g|^3 \, .
\end{equation}
Define \(
\delta(g)\coloneqq C |g|^3 \, .
\)
Combining \eqref{eq:appii_stage_residual} and \eqref{eq:appii_continuation_residual}, we get
\begin{equation}\label{eq:appii_bellman_residual}
\left|\ell(g,\lambda)+\beta\E[\widetilde V(\Psi_\lambda(g)+\varepsilon)]-\widetilde\ell(g,\lambda)-\beta\E[\widetilde V(\widetilde\Psi_\lambda(g)+\varepsilon)]\right|\leq\delta(g) \, .
\end{equation}

The Bellman equation of the auxiliary problem with the constant policy \(\widetilde\lambda\) is
\begin{equation}
\widetilde V(g)=\widetilde\ell(g,\widetilde\lambda)+\beta\E[\widetilde V(\widetilde\Psi_{\widetilde\lambda}(g)+\varepsilon)] \, ,
\end{equation} and applying \eqref{eq:appii_bellman_residual} at \(\widetilde\lambda\) gives
\begin{equation}\label{eq:appii_approximated_policy_residual}
\widetilde V(g)\geq\ell(g,\widetilde\lambda)+\beta\E[\widetilde V(\Psi_{\widetilde\lambda}(g)+\varepsilon)]-\delta(g) \, .
\end{equation}
Let \(g_n^{\widetilde\lambda}\) be the sequence of states at each step under the constant policy \(\widetilde\lambda\) with exact evolution \( g_{n+1} =\Psi_{\lambda_n}(g_n)+\varepsilon_{n+1} \).
Iterating \eqref{eq:appii_approximated_policy_residual} for \(N\) periods gives
\begin{equation}
\E\!\left[\sum_{n=0}^{N-1}\beta^n\ell(g_n^{\widetilde\lambda},\widetilde\lambda)\right]\leq\widetilde V(g)-\beta^N\E[\widetilde V(g_N^{\widetilde\lambda})]+\E\!\left[\sum_{n=0}^{N-1}\beta^n\delta(g_n^{\widetilde\lambda})\right] \, .
\end{equation}
Since \(\widetilde V\geq0\), dropping the terminal term and letting \(N\to\infty\) yields
\begin{equation}\label{eq:appii_exact_cost_upper}
J(g;\widetilde\lambda)\leq\widetilde V(g)+\E\!\left[\sum_{n=0}^{\infty}\beta^n\delta(g_n^{\widetilde\lambda})\right] \, .
\end{equation}

We now proceed to bound \(\widetilde V \) by \(V\).
For any admissible action \(\lambda\), the Bellman inequality gives
\begin{equation}
\widetilde V(g)\leq\widetilde\ell(g,\lambda)+\beta\E[\widetilde V(\widetilde\Psi_\lambda(g)+\varepsilon)] \, .
\end{equation}
Combining this inequality with \eqref{eq:appii_bellman_residual} gives
\begin{equation}\label{eq:appii_exact_policy_residual}
\widetilde V(g)\leq\ell(g,\lambda)+\beta\E[\widetilde V(\Psi_\lambda(g)+\varepsilon)]+\delta(g) \, .
\end{equation}
Fix \(\eta>0\) and choose an \(\eta\)-optimal policy \(\lambda^\eta\in\Lambda\) satisfying \(J(g;\lambda^\eta)\leq V(g)+\eta\).
Iterating \eqref{eq:appii_exact_policy_residual} along its exact state process and letting the horizon tend to infinity gives
\begin{equation}
\widetilde V(g)\leq J(g;\lambda^\eta)+\E\!\left[\sum_{n=0}^{\infty}\beta^n\delta(g_n^{\lambda^\eta})\right] \, .
\end{equation}
Letting \(\eta\downarrow0\) yields
\begin{equation}\label{eq:appii_approximated_value_upper}
\widetilde V(g)\leq V(g)+\sup_{\lambda\in\Lambda}\E\!\left[\sum_{n=0}^{\infty}\beta^n\delta(g_n^\lambda)\right] \, .
\end{equation}
Note that the discounted terminal terms vanish because \(\widetilde V\) is quadratic with bounded coefficients and the exact dynamics have uniformly bounded second moments (possible \(\lambda\) is bounded from below by \(\underline \lambda > 0\)).

We now bound the infinite sum of residuals to conclude the proof.
Let \(g_n^\lambda\) be the exact state under an arbitrary admissible policy \(\lambda\).
The derivative estimate in \eqref{eq:appii_psi_derivative_bound} and \(\lambda_n\geq\underline\lambda\) imply
\begin{equation}
|\Psi_{\lambda_n}(g_n^\lambda)|\leq\kappa|g_n^\lambda|, \quad \kappa\coloneqq1-\underline\lambda\min\{\theta,1-\theta\}<1 \, .
\end{equation}
For the second moment,
\begin{equation}
\E[(g_{n+1}^\lambda)^2]=\Psi_{\lambda_n}(g_n^\lambda)^2+\sigma^2\Delta t\leq\kappa^2(g_n^\lambda)^2+\sigma^2\Delta t \, .
\end{equation}
Iteration gives
\begin{equation}\label{eq:appii_second_moment_bound}
\E[(g_n^\lambda)^2]\leq C\left(\kappa^{2n}g^2+\Delta t\right) \, .
\end{equation}
For the third moment, choose \(\xi>0\) such that \(\alpha_3\coloneqq(1+\xi)\kappa^3<1\).
Using \(|x+y|^3\leq(1+\xi)|x|^3+C_{3,\xi}|y|^3\),
\begin{equation}
\E[|g_{n+1}^\lambda|^3]\leq\alpha_3|g_n^\lambda|^3+C_3 \Delta t^{3/2} \, .
\end{equation}
Therefore
\begin{equation}\label{eq:appii_higher_moment_bound}
\E[|g_n^\lambda|^3]\leq C\left(\alpha_3^n|g|^3+\Delta t^{3/2}\right) \, .
\end{equation}
Since \(\delta(g)=C(|g|^3)\), \eqref{eq:appii_higher_moment_bound} implies
\begin{equation}
\E[\delta(g_n^\lambda)]\leq C\left(\alpha_3^n|g|^3+\Delta t^{3/2}\right) \, .
\end{equation}
Summing with discount factor \(\beta\) gives
\begin{equation}\label{eq:appii_discounted_residual_bound}
\sup_{\lambda\in\Lambda}\E\!\left[\sum_{n=0}^{\infty}\beta^n\delta(g_n^\lambda)\right]\leq C\left(|g|^3+\sqrt{\Delta t}\right) \, ,
\end{equation}
because \((1-\beta)^{-1}=\bigO((\Delta t)^{-1})\).

Applying \eqref{eq:appii_discounted_residual_bound} to \eqref{eq:appii_exact_cost_upper} and \eqref{eq:appii_approximated_value_upper} gives
\begin{equation}\label{eq:appii_value_comparisons}
J(g;\widetilde\lambda)\leq\widetilde V(g)+R(g), \qquad \widetilde V(g)\leq V(g)+R(g), \qquad R(g)\coloneqq C\left(|g|^3+\sqrt{\Delta t}\right) \, .
\end{equation}
Since \(V(g)\leq J(g;\widetilde\lambda)\) by definition, \eqref{eq:appii_value_comparisons} implies
\begin{equation}
|V(g)-\widetilde V(g)|\leq R(g), \qquad |J(g;\widetilde\lambda)-\widetilde V(g)|\leq R(g), \qquad 0\leq J(g;\widetilde\lambda)-V(g)\leq2R(g) \, .
\end{equation}
Absorbing the factor \(2\) into \(C\) proves \eqref{eq:optiii_near_optimality_global}, and for \(|g|\leq M\sqrt{\Delta t}\), \(|g|^3=\bigO\left(\Delta t^{3/2}\right)\), which proves \eqref{eq:optiii_near_optimality_local}.
Finally, \cref{thm:optiii_approximated_solution} gives, uniformly over this local regime,
\begin{equation}
\widetilde V(g)=\widetilde v_2g^2+\frac{\beta}{1-\beta}\widetilde v_2\sigma^2\Delta t=\frac{\sigma^2}{\rho}\widetilde v+\bigO(\Delta t) \, .
\end{equation}
Combining this with \(|V-\widetilde V|+|J-\widetilde V|=\bigO(\sqrt{\Delta t})\) proves \eqref{eq:optiii_value_scaling} while \(\frac{\sigma^2}{\rho} \widetilde v > 0\).
Therefore both \( J \) and \(V\) are of \(\Theta(1)\), and division by \(V\) proves \eqref{eq:optiii_relative_near_optimality}.
\end{proof}



\subsection{Pseudocode of PA-AMM implementation}

\begin{algorithm}[H]
  \small
  \label{alg:pa_amm_logic_ii}
  \begin{algorithmic}[1]
    \State \textbf{Data:}
    \State \quad $R_{\texttt{active}}, R_{\texttt{passive}} \in \mathbb{R}_+^2$ \Comment{Current active and passive reserves}
    \State \quad $n_{\texttt{last}} \in \mathbb{N}$ \Comment{Block number of last rebalance}
    \State \quad $\lambda \in (0, 1]$ \Comment{Activeness parameter}
    \State \quad $s_\texttt{total}$ \Comment{Total supply of LP tokens}

    \vspace{0.2cm}

    \Procedure{Rebalance}{}
    \If{block.number $> n_{\texttt{last}}$} \Comment{Check if this is the first interaction}
    \State $R_{\texttt{total}} \gets R_{\texttt{active}} + R_{\texttt{passive}}$
    \State $R_{\texttt{active}} \gets \lambda \cdot R_{\texttt{total}}$ \Comment{Reset active reserves to $\lambda$ portion of total reserves}
    \State $R_{\texttt{passive}} \gets (1 - \lambda) \cdot R_{\texttt{total}}$ \Comment{Keep the remaining liquidity idle}
    \State $n_{\texttt{last}} \gets \text{block.number}$
    \EndIf
    \EndProcedure

    \vspace{0.2cm}

    \Procedure{Swap}{$\Delta x, \Delta y$} \Comment{Execute a trade $\Delta x, \Delta y$}
    \State \Call{Rebalance}{}
    \State \Assert $\varphi(R_\texttt{active} + (\Delta x, \Delta y)) \geq \varphi(R_\texttt{active})$ \Comment{Check preservation of the invariant}
    \State $R_{\texttt{active}}^x \gets R_{\texttt{active}}^x + \Delta x$ \Comment{Update the base-asset reserve}
    \State $R_{\texttt{active}}^y \gets R_{\texttt{active}}^y + \Delta y$ \Comment{Update the quote-asset reserve}
    \EndProcedure

    \vspace{0.2cm}

    \Procedure{AddLiquidity}{$s_\texttt{add}$}
    \State \Call{Rebalance}{}
    \State $R_\texttt{active} \gets R_\texttt{active}(1 + s_\texttt{add}/s_\texttt{total})$
    \State $R_\texttt{passive} \gets R_\texttt{passive}(1 + s_\texttt{add}/s_\texttt{total})$
    \State $s_\texttt{total} \gets s_\texttt{total} + s_\texttt{add}$
    \EndProcedure

    \vspace{0.2cm}

    \Procedure{RemoveLiquidity}{$s_\texttt{remove}$}
    \State \Call{Rebalance}{}
    \State $R_\texttt{active} \gets R_\texttt{active}(1 - s_\texttt{remove}/s_\texttt{total})$
    \State $R_\texttt{passive} \gets R_\texttt{passive}(1 - s_\texttt{remove}/s_\texttt{total})$
    \State $s_\texttt{total} \gets s_\texttt{total} - s_\texttt{remove}$
    \EndProcedure
  \end{algorithmic}
  \caption{Rebalancing, swaps, additions, and removals of liquidity in PA-AMM.}
\end{algorithm}